\documentclass[11pt]{article}
\pdfoutput=1
\usepackage{wrapfig}
\usepackage{fullpage}           % layout
\usepackage{amsfonts}           % \mathbb
\usepackage{amsthm}             % proof
\usepackage{xspace}             % \xspace
\usepackage{graphicx}            % \includegraphics
\usepackage{adjustbox}          % \adjustbox 
\usepackage{multicol}            
\usepackage{url}
\usepackage{enumerate}  %enumerate environment with optional argument
\usepackage{subcaption}
\usepackage[usenames]{xcolor} % for coordinating edits
\usepackage{soul}             % \st

\colorlet{darkgreen}{green!45!black}

\newcommand{\R}{\mathbb{R}}

\newcommand{\etal}{et al.\xspace}

\newtheorem{lemma}{Lemma}

\usepackage{amsmath}
\begin{document}
\title{Balanced power diagrams for redistricting}
\author{Vincent Cohen-Addad\thanks{CNRS, UPMC Paris}\and
  Philip N. Klein\thanks{Brown University, Research supported by
    National Science Foundation Grant CCF-1409520.} \and Neal
  E. Young\thanks{University of California, Riverside.  Research
    supported by NSF Grant IIS-1619463}}
\date{January 6, 2018}
\maketitle 

\begin{abstract}
We explore a method for \emph{redistricting}, decomposing a  
% We propose a method for \emph{redistricting}, decomposing a  
geographical area into subareas, called \emph{districts}, so that the
populations of the districts are as close as possible and the
districts are compact and contiguous.  Each district is the
intersection of a polygon with the geographical area.  The polygons
are convex and the average number of sides per polygon is less than
six.  The polygons tend to be quite compact.  With each polygon is
associated a \emph{center}.  The center is the centroid of the
locations of the residents associated with the polygon.  The algorithm
can be viewed as a heuristic for finding centers and a balanced assignment of
residents to centers so as to minimize the sum of squared distances of
residents to centers; hence the solution can be said to have low
dispersion.
\end{abstract}

\section{Introduction}\label{sec: intro}
% We observe that the optimization problem \emph{capacitated $k$-means clustering} is
% relevant to a real-world problem, \emph{redistricting}.
\paragraph{Redistricting.}

\emph{Redistricting},
in the context of elections
refers to decomposing a geographical area into subareas such that all
subareas have the same population.  The subareas are called 
\emph{districts}.  In most US states, districts are supposed to be 
\emph{contiguous} to the extent that is possible.  
Contiguous can reasonably be interpreted to mean 
\emph{connected}.  

In most states, districts are also supposed to be 
\emph{compact}.  This is not precisely defined in law.  
Some measures of
compactness are based on boundaries; a district is preferred if its
boundaries are simpler rather than contorted.  Some measures are based
on \emph{dispersion}, ``the degree to which the district spreads from a
central core''~\cite{Levitt}.
Idaho directs its redistricting commision
to ``avoid drawing districts that are oddly shaped.''  Other states
loosely address the meaning of compactness: ``Arizona and
Colorado focus on contorted boundaries; California, Michigan, and
Montana focus on dispersion; and Iowa embraces both''~\cite{Levitt}.

\paragraph{Balanced centroidal power diagrams}\label{sec: intro def}

The goal of this paper is to explore a particular approach to redistricting: 
% The goal of this paper is to propose a particular approach to redistricting: 
\emph{balanced centroidal power diagrams}.  Given the locations
of a state's $m$ residents and given the desired number $k$ of
districts, a balanced centroidal power diagram
partitions the state into $k$ districts with the following desirable properties:
\begin{description}
\item[(P1)] each district is the intersection of the state with a convex polygon,\label{prop:polygon}
\item[(P2)] the average number of sides per polygon is less than six, and
\item[(P3)] the populations of the districts differ by at most one.\label{prop:population}
\end{description}

A balanced centroidal power diagram is a particular kind of
(not necessarily optimal) solution to an optimization problem called
\emph{balanced $k$-means clustering}:  given a set $P$ of $m$ points (the \emph{residents})
and the desired number $k$ of clusters, a
solution (not necessarily of minimum cost) consists of a sequence $C$ of $k$ points
(the \emph{centers}) and an assignment $f$ of
residents to centers that is \emph{balanced}:
it assigns $\lfloor m/k\rfloor$ residents to the first $i$ centers,
and $\lceil m/k\rceil$ residents to the remaining $k-i$ centers
(for the $i$ such that $i \lfloor m/k\rfloor + (k-i) \lceil m/k\rceil = m$).
The \emph{cost} of a
solution $(C,f)$ is the sum, over the residents,
of the square of the Euclidean distance between the resident's location and assigned center.
(This is a natural measure of dispersion.)
In \emph{balanced $k$-means clustering}, one seeks a solution of minimum cost. 
This problem is NP-hard~\cite{mahajan2009planar}.

A balanced centroidal power diagram arises from a solution to balanced
$k$-means clustering that is not necessarily of minimum cost.
Instead, the solution $(C, f)$ only needs to be a \emph{local minimum},
meaning that it is not possible to lower the cost by just varying $f$
(leaving $C$ fixed), or just varying $C$ (leaving $f$ fixed).
Local minima tend to have low cost, so tend to have low dispersion.

\smallskip

Section~\ref{sec: history} reviews the meaning of the
terms \emph{centroidal} and \emph{power diagram},
and discusses how any such local minimum yields districts
(with each district containing the residents assigned to one center)
for which the desirable properties~(P1)--(P3) 
are mathematically guaranteed.
Convex polygons with few sides are arguably well shaped,
and their boundaries are arguably not contorted.
% We propose an algorithm for finding such
% decompositions that tends to find solutions that \emph{also} have very small
% dispersion according to an intuitive measure we describe shortly.

The idea and its application to redistricting are not novel.  Spann et
al.~\cite{spann_electoral_2007} describes a method to find a
centroidal power diagram that is nearly balanced (to within 2\%).
Their solutions are thus not exactly balanced in the sense we have
defined.  We discuss this and other related work
in Section~\ref{sec: history}.

Figures~\ref{fig:FL} to~\ref{fig:LI}
show proposed districts corresponding to balanced centroidal power diagrams
for the six most populous states in the U.S,
based on population data from the 2010 census
(locating each resident at the centroid of that resident's census
block).  We will also show such diagrams at a web site, ~\url{district.cs.brown.edu}.

We computed these diagrams efficiently using a variant of Lloyd's algorithm:
start with a random set $C$ of centers,\footnote
{The probability distribution we used for the initial set of centers is from~\cite{ArthurV07}.}
then repeat the following steps until an equilibrium is reached:
(1) given the current set $C$ of centers, compute a balanced assignment $f$ that minimizes the cost;
(2) given that assignment $f$, change the locations of the centers in $C$ so as to minimize the cost.

\newcommand{\imagefigure}[3]
{\begin{figure}\centering
    \includegraphics[trim=2.8cm 0cm 2.6cm 0cm,width=\textwidth]{Images/#1}
  \caption{#2}\label{#3}
\end{figure}}

\imagefigure{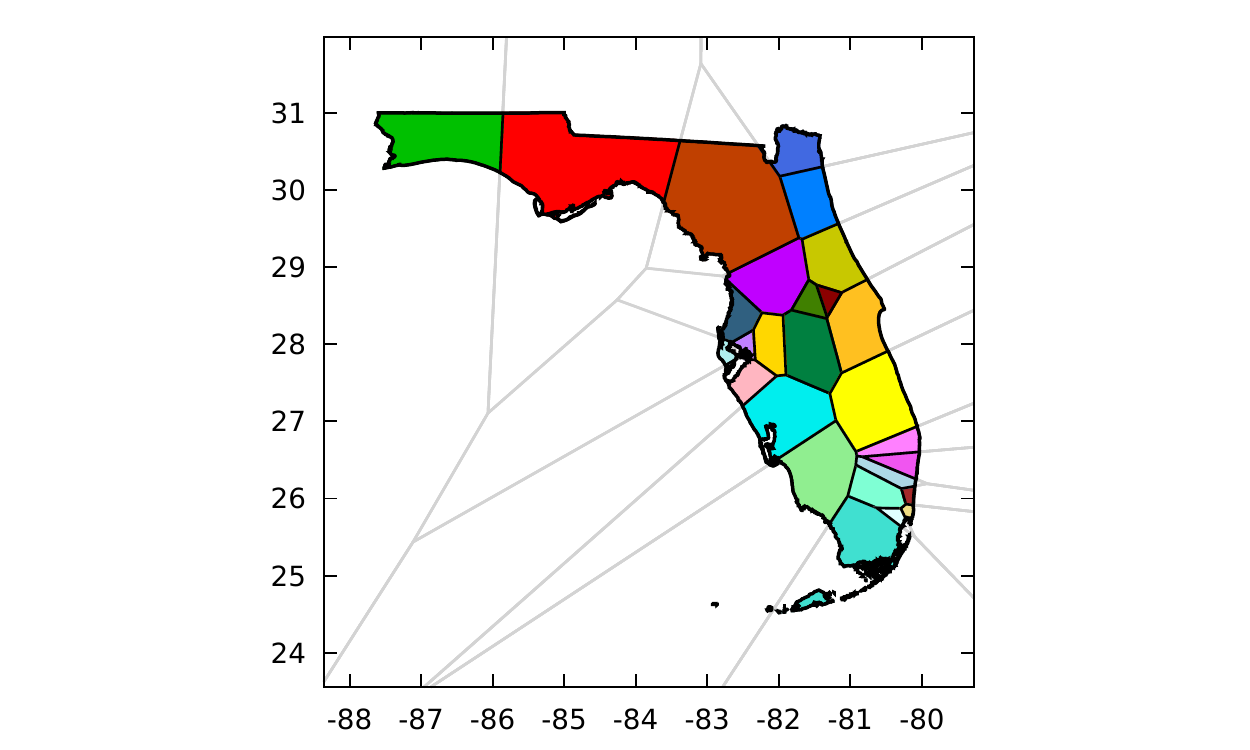}{Florida (27 districts)}{fig:FL}

\imagefigure{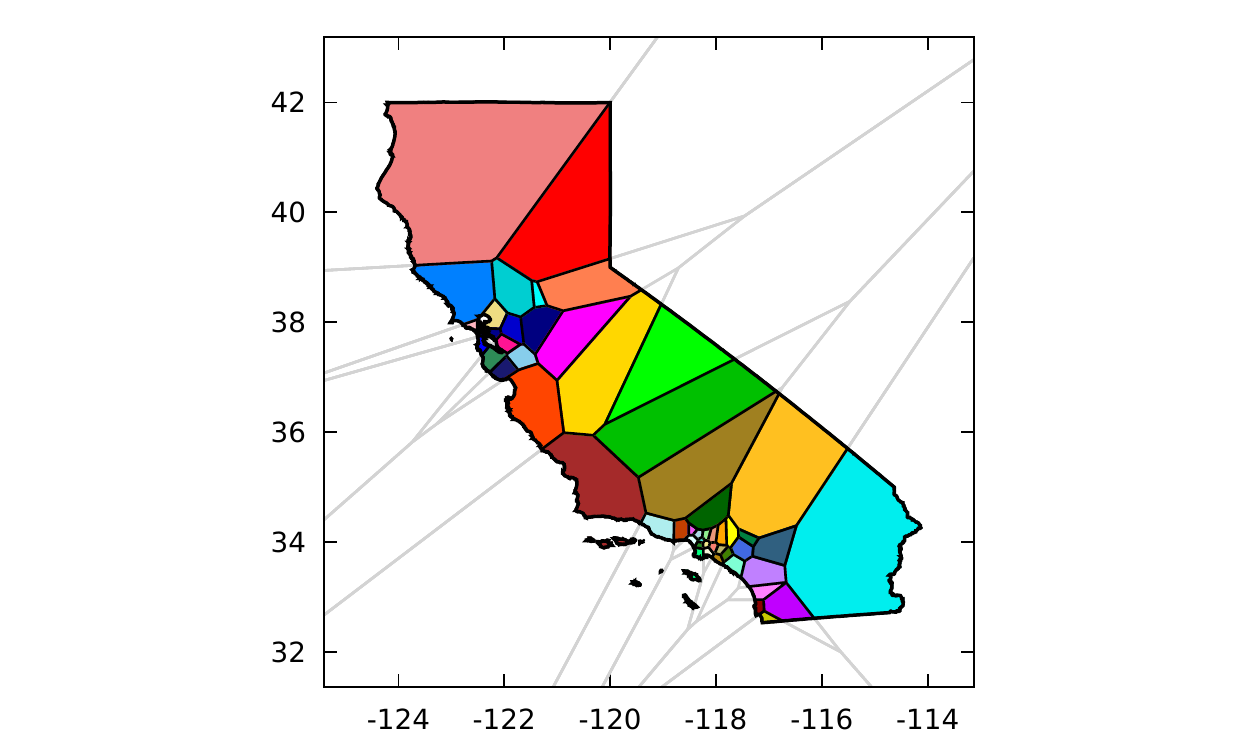}{California (53 districts).}{fig:CA}

\imagefigure{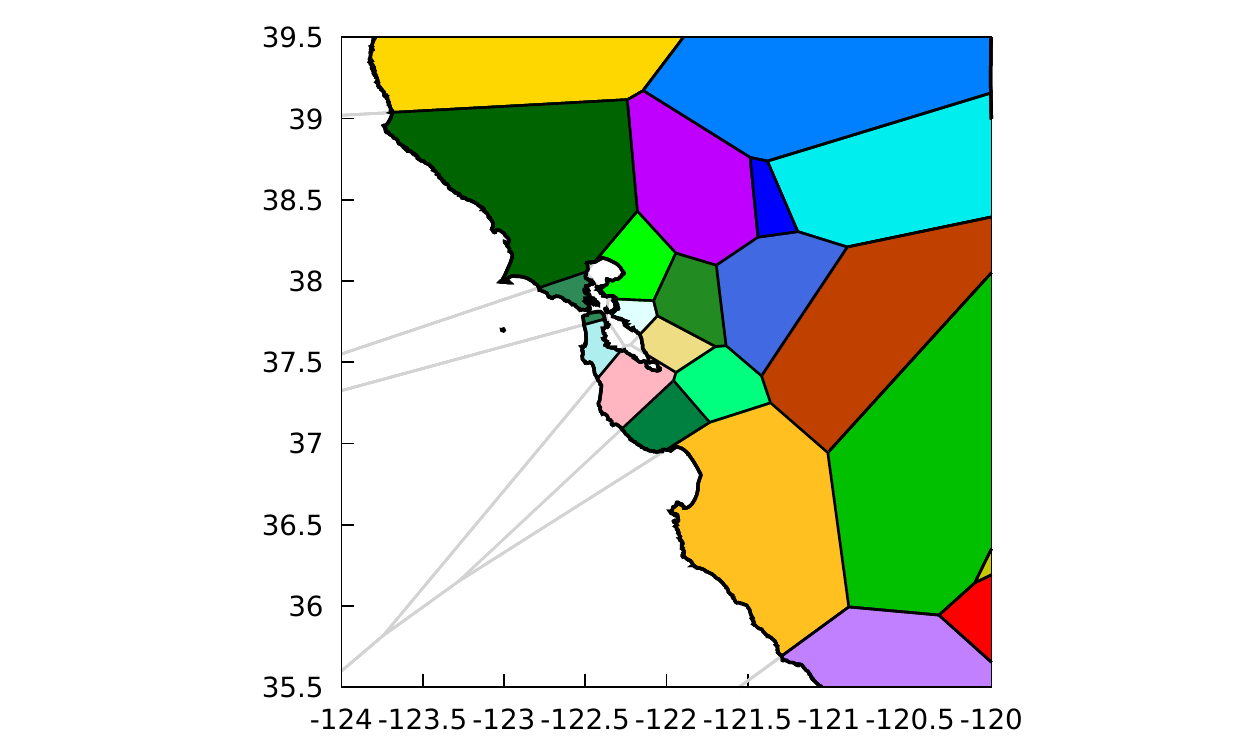}{Bay Area (detail of \emph{California}).}{fig:bay}

\imagefigure{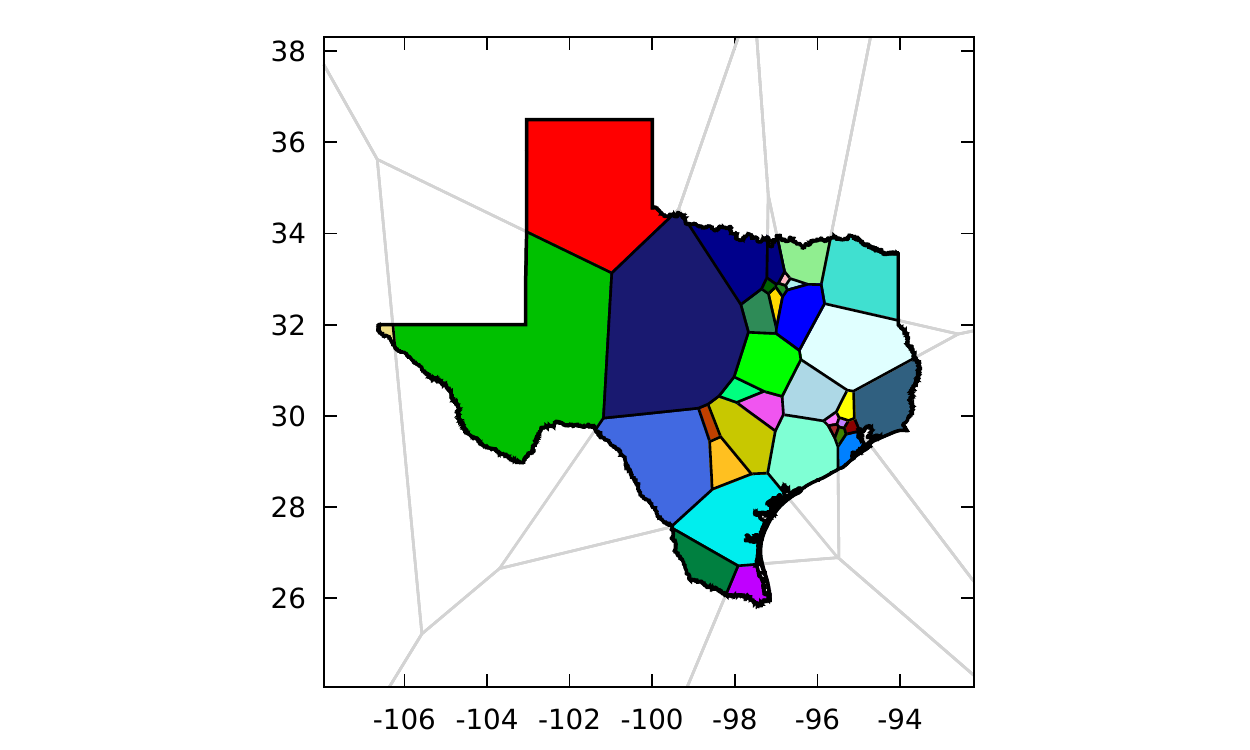}{Texas (36 districts).}{fig:TX}

\imagefigure{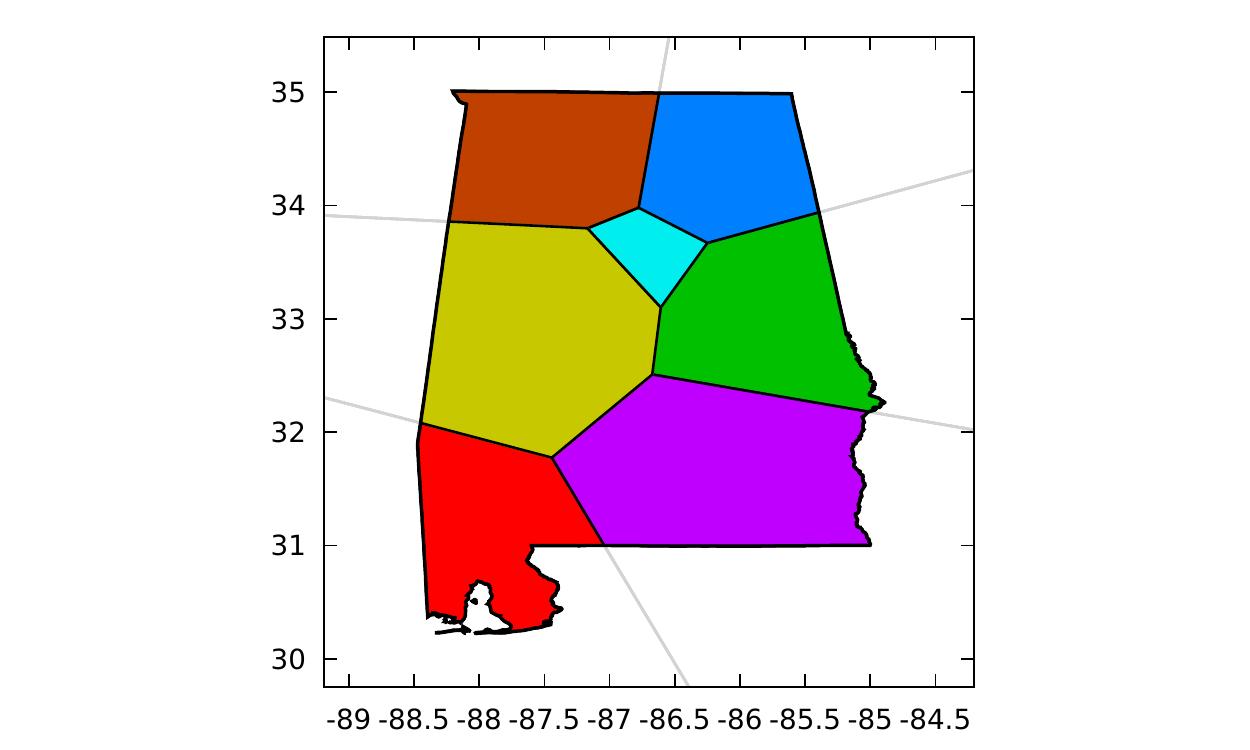}{Alabama (7 districts).}{fig:AL}

\imagefigure{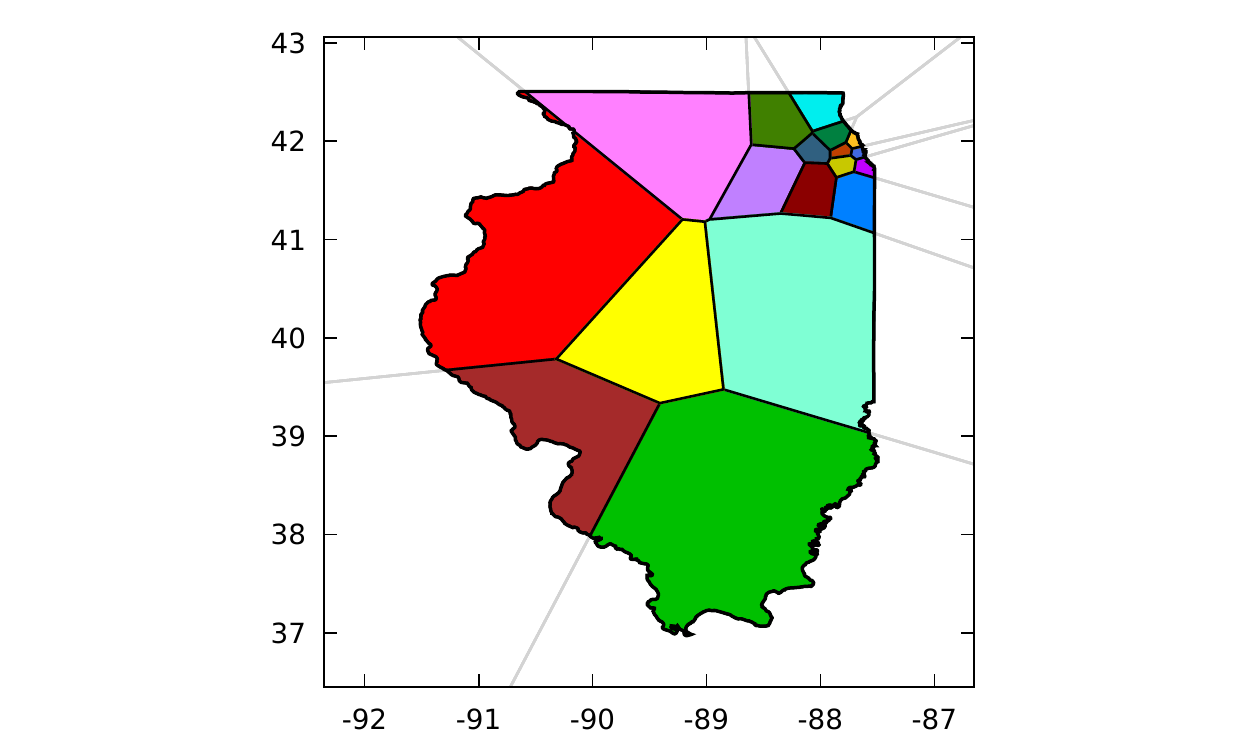}{Illinois (18 districts).}{fig:IL}

\imagefigure{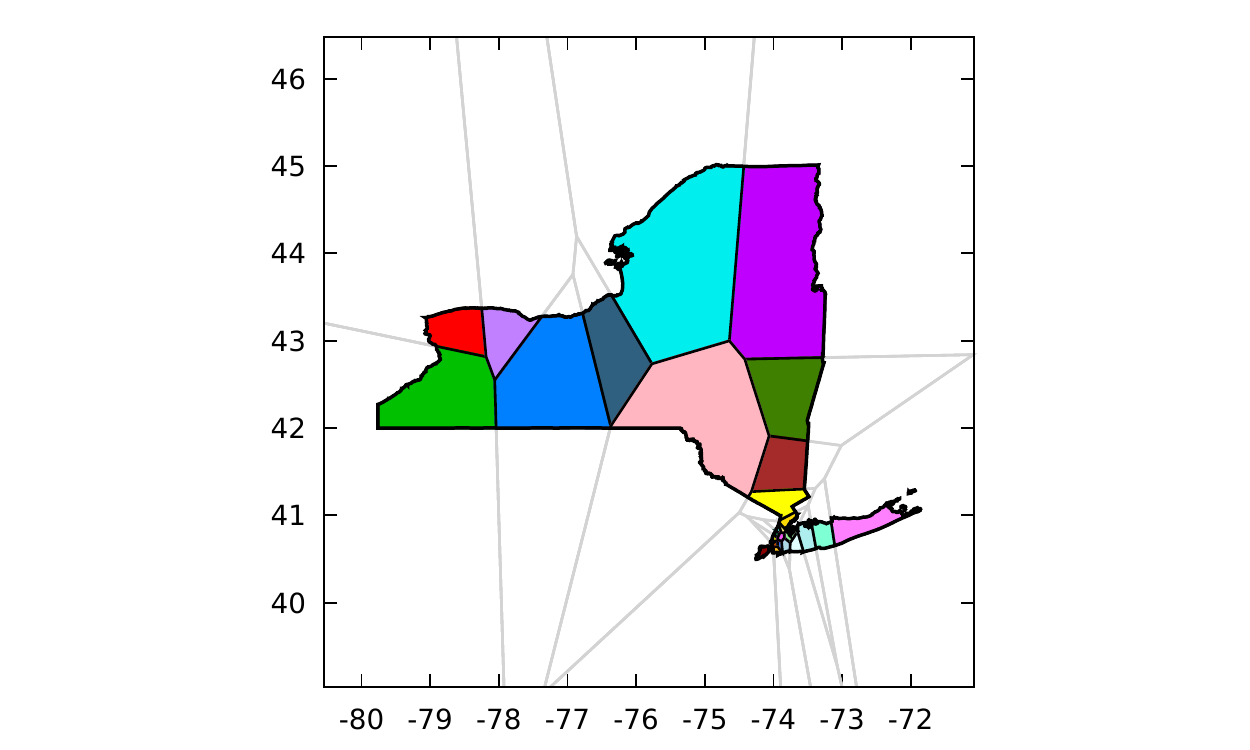}{New York (27 districts).}{fig:NY}

\imagefigure{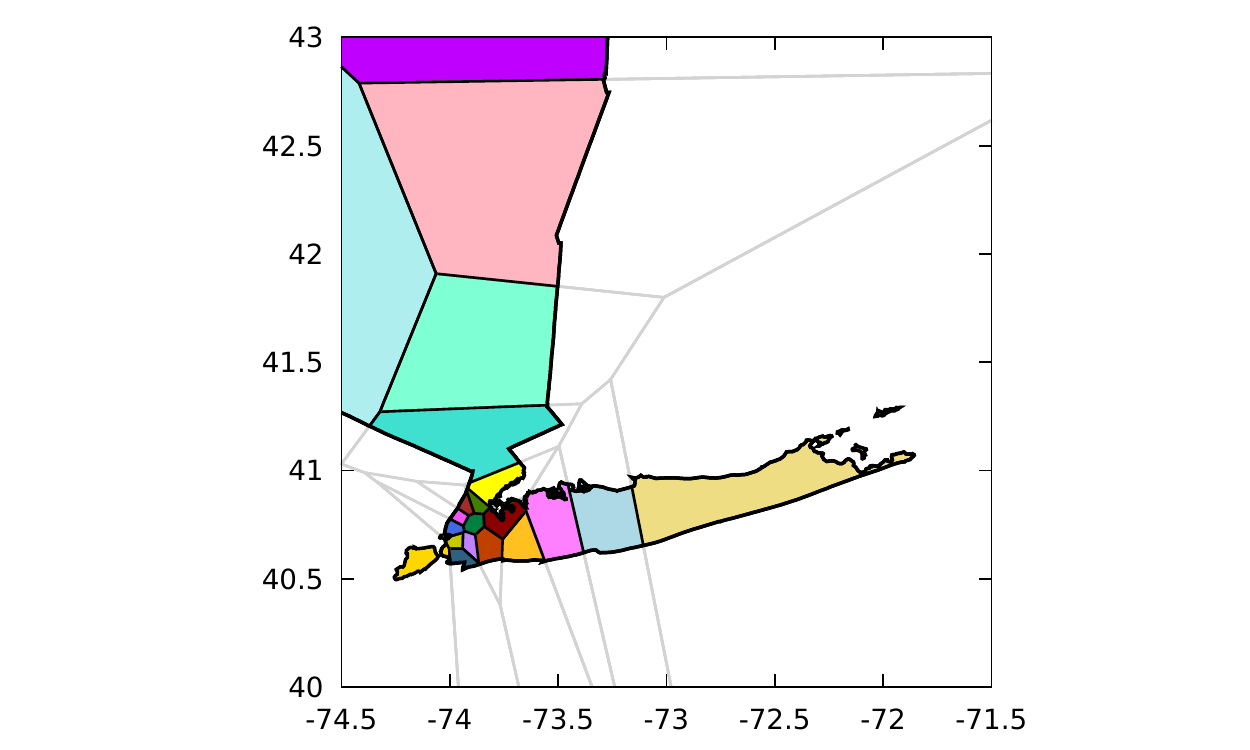}{Long Island, New York and Manhattan
  (detail from  \emph{New York}).}{fig:LI}

Some might object that the method does not provide the scope for
achieving some other goals, e.g.\ creating competitive districts.  A
counterargument is that one should \emph{avoid} providing politically
motivated legislators the scope to select boundaries of districts so
as to advance political goals.  According to this argument, the less
freedom to influence the district boundaries, the better.  This method
does not guarantee fairness in outcome; the fairness is in the
process.  This point was made, e.g., by Miller~\cite{miller_problem_2007}.

Note that in real applications of redistricting, the locations of
people are not given precisely.  Rather, there are regions, called
\emph{census blocks}, and each such region's population is specified.

%%% Local Variables:
%%% TeX-master: "equitable"
%%% Mode: latex
%%% End:

\section{Balanced centroidal power diagrams}\label{sec: history}

The use of optimization, generally, for redistricting has been
proposed starting at least as far back as 1965 and has continued up to the
present~\cite{hess,garfinkel_optimal_1970,eppstein_defining_2017}.
See~\cite{altman2010promise,olson_rangevoting.org_2011} for additional references.

% \textcolor{red}{
%   \sout{We note especially the work of Helbig, Orr, and
%     Roediger}~\cite{helbig1972political}\sout{ in 1972, who proposed an iterative
%     method based on the \emph{transportation problem}.  We will discuss
%     their method in greater detail later.}
% }

In what follows, we focus specifically on the use of balanced centroidal power diagrams.
Next is a summary of the relevant history, interspersed with necessary definitions.
Throughout,
$P$ (the \emph{population}) denotes a set of $m$ \emph{residents} (points in a Euclidean space),
$C$ denotes a sequence of $k$ \emph{centers} (points in the same space),
$f:P\rightarrow C$ denotes an assignment of residents to centers,
and $d(y, x)$ denotes the distance from $y\in P$ to $x\in C$.
We generally consider the parameters $P$ and $k$ to be fixed throughout,
while $C$ and $f$ vary.

\newcommand{\power}[2]{{\cal P}(#1, #2)}
\newcommand{\Power}[3]{{\cal P}(#1, #2, #3)}
\newcommand{\voronoi}[1]{{\cal V}(#1)}
\newcommand{\Voronoi}[2]{{\cal V}(#1, #2)}

\paragraph{The power diagram of $(C, w)$.} 
Given any sequence $C$ of centers, and a weight $w_x\in\R$ for each center $x\in C$,
the \emph{power diagram} of $(C, w)$, denoted $\power C w$, is defined as follows.
For any center $x\in C$, 
the \emph{weighted squared distance} from any point $y$ to $x$ is $d^2(y,x) - w_x$.
The \emph{power region} $C_x$ associated with $x$ consists of all points
whose weighted squared distance to $x$ is no more than the weighted squared distance to any other center. 
The power diagram $\power C w$ is the collection of these power regions.

An assignment $f: P \rightarrow C$ is \emph{consistent} with $\power C w$
if every resident assigned to center $x$ belongs to the corresponding region $C_x$.
(Residents in the interior of $C_x$ are necessarily assigned to $x$.)
$\Power C w f$ denotes the power diagram $\power C w$ augmented with such an assignment.

Power diagrams are well-studied~\cite{aurenhammer_power_1987}. 
If the Euclidean space is $\R^2$, it is known that each power region $C_x$ is necessarily a (possibly infinite) convex polygon.
If each weight $w_x$ is zero, the power diagram is also called a \emph{Voronoi diagram},
and denoted $\voronoi C$.
Likewise $\Voronoi C f$ denotes the Voronoi diagram extended with a consistent assignment $f$
(which simply assigns each resident to a nearest center).

\paragraph{Centroidal power diagrams.} 
A \emph{centroidal power diagram} is an augmented power diagram $\Power C w f$
such that the assignment $f$ is \emph{centroidal}:
each center $x\in C$ is the centroid (center of mass) of its assigned residents, $\{y\in P: x = f(y)\}$.

\paragraph{Centroidal Voronoi diagrams.} 
Centroidal Voronoi diagrams (a special case of centroidal power diagrams)
have many applications~\cite{du_centroidal_1999}.
One canonical application from graphics is downsampling a given image,
by partitioning the image into regions, then selecting a single pixel from each region to represent the region.
Centroidal Voronoi diagrams are preferred over arbitrary Voronoi diagrams
because the regions in centroidal Voronoi diagrams tend to be more compact.

\smallskip

\emph{Lloyd's method} is a standard way to compute a centroidal Voronoi diagram $\Voronoi C f$,
given $P$ and the desired number of centers, $k$~\cite
[\S\,5.2]{du_centroidal_1999}.
Starting with a sequence $C$ of $k$ randomly chosen centers, 
the method repeats the following steps until the steps do not cause a
change in $f$ or $C$:
\begin{enumerate}
\item Given $C$, let $f$ be any assignment assigning each resident to a nearest center in $C$.
\item Move each center $x\in C$ to the centroid of the residents that $f$ assigns to $x$.
\end{enumerate}

Recall that the \emph{cost} is $\sum_{y\in P} d^2(y, f(y))$. 
Step (1) chooses an $f$ of minimum cost, given $C$.
Step (2) moves the centers to minimize the cost, given $f$.
Each iteration except the last reduces the cost, so the algorithm terminates
and, at termination, $(C,f)$ is a \emph{local minimum} in the following sense:
by just moving centers in $C$, or just changing $f$, it is not possible to reduce the cost.

In the last iteration,
Step (1) computes $f$ that is consistent with $\voronoi C$,
and Step (2) does not change $C$, so $f$ is centroidal.
So, at termination, $\Voronoi C f$ is the desired centroidal Voronoi diagram.

\smallskip

Miller~\cite{miller_problem_2007} and Kleiner \etal~\cite{kleiner_political_2013}
explore the use of centroidal Voronoi diagrams specifically for \emph{redistricting}.
The resulting districts (regions) are guaranteed to be polygonal,
and tend to be compact, but their populations can be far from balanced.
To address this, %\textcolor{red}{following~\cite{spann_electoral_2007}},
consider instead \emph{balanced} centroidal power diagrams, described next,
% We compute them using a capacitated variant of Lloyd's method.
%\textcolor{red}
{which can be computed using} a capacitated variant of Lloyd's method.

\paragraph{Balanced power diagrams.}
A \emph{balanced} power diagram is an augmented power diagram $\Power C w f$
such that the assignment $f$ is balanced (as defined in the introduction).
Hence, the numbers of residents in the regions of $\power C w$ differ by at most 1.
Such regions are desirable in many applications.

Aurenhammer \etal~\cite[Theorem~1]{aurenhammer_minkowski-type_1998}
give an algorithm that, in the case of a Euclidean metric, given $P$ and $C$, 
computes weights $w$ and an assignment $f$
such that $\Power C w f$ is a balanced power diagram,
and $f$ has minimum cost among all balanced assignments of $P$ to $C$.
%\textcolor{red}{By a duality argument similar to one in~\cite{spann_electoral_2007}},
We observe in Section~\ref{duality} that, given $P$, $C$,
there exist weights $w$ such that $\Power C w f$ is a balanced power diagram
for \emph{any} minimum-cost balanced assignment $f$ and any metric.
Such an argument was previously presented by Spann et al.~\cite{spann_electoral_2007}.

\paragraph{Computing a balanced centroidal power diagram for $P$.} 
A \emph{balanced centroidal power diagram} is an augmented power diagram $\Power C w f$
such that $f$ is both balanced and centroidal.
We %\textcolor{red}
{implement} the following capacitated variant of Lloyd's method to compute such a diagram,
given $P$ and the desired number $k$ of centers.
Starting with a sequence $C$ of $k$ randomly chosen centers,
repeat the following steps until Step (2) doesn't change $C$:
\begin{enumerate}
\item Given $C$, compute a minimum-cost balanced assignment $f:P\rightarrow C$.
\item Move each center $x\in C$ to the centroid of the residents that $f$ assigns to it.
\end{enumerate}

As in the analysis of the uncapacitated method,
each iteration except the last reduces the cost, $\sum_{y\in P} d^2(y, f(y))$,
and at termination, the pair $(C, f)$ is a local minimum in the following sense:
by just moving the centers in $C$,
or just changing $f$ (while respecting the balance constraint), it is not possible to reduce the cost.

\smallskip

The problem in Step (1) can be solved via Aurenhammer \etal's algorithm, described previously.
Instead, as described in Section~\ref{duality},
we solve it by reducing it to minimum-cost flow;
yielding both the stipulated $f$
and (via the dual variables) weights $w$ such that $\Power C w f$ is a
balanced power diagram.  Note that the solution obtained by
minimum-cost flow assigns assigns each person to a single district. 
In the last iteration,  
Step (2) does not change $C$, so $f$ is also centroidal,
and at termination $\Power C w f$ is a balanced centroidal power
diagram, as desired. 

\smallskip

In previous work, Spann \etal~\cite{spann_electoral_2007} proposed a
similar iterative method to find a centroidal power diagram.  They did
not seem to be aware of the work of Aurenhammer
\cite{aurenhammer_minkowski-type_1998} but used a duality argument to
derive power weights.  It is
not clear from their paper precisely how their method carries out
Step~(1).  They state that their implementation starts by allowing a
20\% deviation from balance, and iteratively reduces the allowed
deviation over a series of iterations, adjusting the target
populations per district, and terminates when the deviation is within
2\% of balanced.  We believe that the additional complexity and the failure
to achieve perfect balance is a result of the authors' effort to
ensure that census blocks are not split.

Hess et al.~\cite{hess} had previously given a similar method.  Like
that of Spann et al., this method ensures that census enumeration
districts (analogous to census blocks) were not split, and as a
consequence did not achieve perfect balance.  Unlike the method of
Spann et al., the method of Hess et
al. did not compute power weights and did not output a power diagram;
presumably each output district is defined as the union of census
enumeration districts and is therefore not guaranteed to be connected.
 
Balzer \etal~\cite{balzer_capacity-constrained_2008,balzer_capacity-constrained_2009}
proposed an algorithm equivalent to {the iterative algorithm above},
except that a local-exchange heuristic (updating $f$ by swapping pairs of residents) was proposed to carry out Step (1).
That heuristic does not guarantee that $f$ has minimum cost (given $C$),
so does not in fact guarantee that the assignment is consistent with a balanced power diagram (see Figure~\ref{fig: Balzer incorrect}).

{Other} previously published algorithms~\cite
{balzer_capacity-constrained_2008,balzer_capacity-constrained_2009,li_fast_2010,de_goes_blue_2012,xin_centroidal_2016}
for balanced centroidal power diagrams address applications (e.g.~in graphs) that have very large instances,
and for which it is not crucial that the power diagrams be exactly centroidal or exactly balanced.
{This class of algorithms prioritize speed, and none are}
guaranteed to find a local minimum $(C, f)$, nor a balanced centroidal power diagram.

%\marginpar{{\footnotesize\raggedright Maybe cut down (or out) the text about Helbig et al? -N}}
Helbig, Orr, and Roediger~\cite{helbig1972political} proposed a somewhat similar
redistricting algorithm.  Like {the algorithm above}, their algorithm
initializes the center locations randomly and then alternates between
(1) using mathematical programming to find an assignment of residents
to centers and (2) replacing each center with the centroid of the
residents assigned to it.

But their
assignment of residents to centers is chosen in each iteration to
minimize the sum of distances, not the sum of squared distances.  This
means that the partition does not correspond to a power diagram.
Indeed, Helbig \etal acknowledge the possibility that noncontiguous
districts could result although they did not observe this occurring.

Their mathematical program for the assignment also
treats each ``population unit'' (e.g.~census block) atomically,
rather than treating each individual that way.  Thus their method never splits a population unit into two districts.
While a solution with this property might be desirable, imposing this
requirement means that a solution might not exist that achieves
population balance.  Moreover, their
mathematical program constrains the \emph{number} of population units
assigned to a center to be a certain number, rather than constraining
the population to be a certain number.  Since different population
units have different populations, this might not achieve population balance.
Helbig \etal address this issue by iteratively modifying the number
of population units to be assigned to each center using a heuristic.
This does not guarantee convergence, so they allow their algorithm to stop
before reaching a true local minimum.

{For an excellent survey of the redistricting algorithms,
  including additional discussion of the Spann \etal algorithm
  and extensions to districts lying on the sphere,
  see the online survey by Olson and Smith~\cite{olson_rangevoting.org_2011}.}

%%% Local Variables:
%%% TeX-master: "equitable"
%%% Mode: latex
%%% End:

\begin{figure}[t]
  \centering
  \includegraphics[height=1.25in]{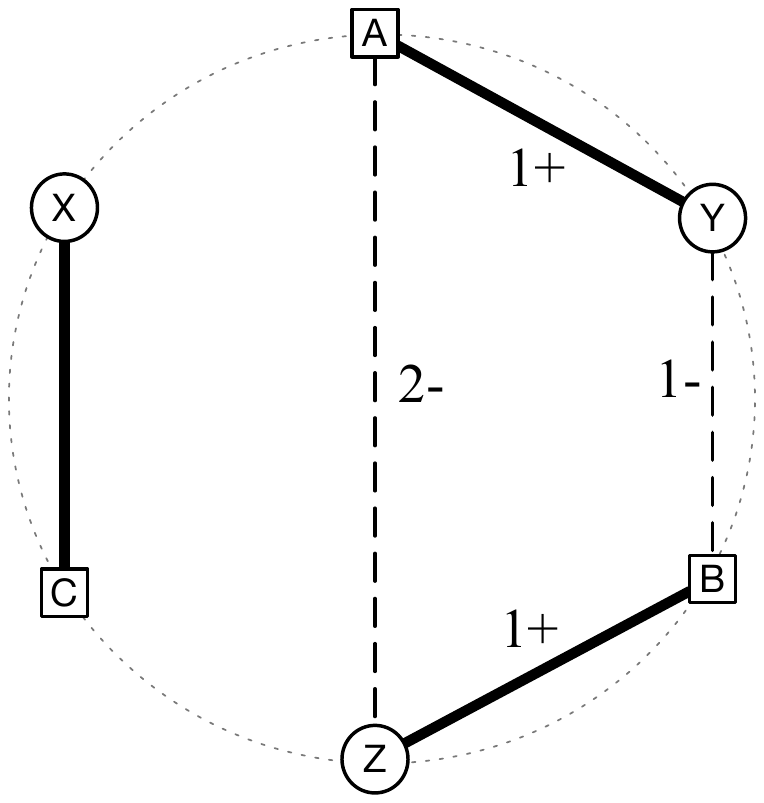}
  \caption{\small A counter-example to the swap-based balanced assignment algorithm of Balzer
    \etal~\cite{balzer_capacity-constrained_2008,balzer_capacity-constrained_2009}.
    The $k=3$ centers $C=\{A,B,C\}$ (each of capacity 1) are the even vertices of a hexagon with unit sides.
    The $m=3$ residents $P=\{X,Y,Z\}$  are the odd vertices.
    The vertices are perturbed slightly so that
    the edges in $M'=\{(A,Y), (B,Z), (C, X)\}$ each have distance slightly more than 1,
    while the edges in $M^*=\{(A,X), (B,Y), (C,Z)\}$ each have distance slightly less than 1,
    making $M^*$ the optimal assignment.
    But if $M'$ is the current assignment,
    for any two sites, say, $A$ and $B$,
    there is only one possible swap,
    and it \emph{increases} the squared distances
    (by about $(2^2+1^1) - (1^2 + 1^1) = 3$).
    So no local improvement is possible from $M'$,
    even though $M'$ does not minimize the sum of the squared distances.
  }\label{fig: Balzer incorrect}
\end{figure}

%%% Local Variables:
%%% mode: latex
%%% TeX-master: "equitable"
%%% End:

\section{An Implementation}\label{sec: details}

\subsection{Minimum-cost flow}\label{duality}
Aurenhammer et al.~\cite{aurenhammer_minkowski-type_1998} provide an
algorithm that, given the set $P$ of locations of residents and the
sequence $C$ of centers, and given a target population for each
center (where the targets sum to the total population), finds a
minimum-cost assignment $f$ of residents to centers 
subject to the constraint that the number of residents assigned to
each center equals the center's target population.  Their algorithm
also outputs weights $w$ for the centers such that the assignment $f$
is consistent with $\power C w$.  Their algorithm can be used to find
a minimum-cost balanced assignment by using appropriate targets.

In the implementation here, we take a different approach to computing the
minimum-cost balanced assignment: we use an algorithm for minimum-cost flow.
Aurenhammer et al.~\cite{aurenhammer_minkowski-type_1998} acknowledge
that a minimum-cost flow algorithm can be used but argue that their
method is more computationally efficient.  As we observe below,
the necessary weights $w$ can be computed from the values of the variables of the
linear-programming dual to minimum-cost flow.
% Note that this approach works in metrics other than Euclidean metrics.

The goal is to find a balanced assignment $f:P\rightarrow C$ of minimum cost, $\sum_{y\in P} d^2(y, p(y))$.
Let $u_x\in\{\lfloor m/k\rfloor, \lceil m/k\rceil\}$ be the number of residents
that $f$ must assign to center $x\in C$.

Consider the following linear program and dual:

\begin{center}
  \begin{tabular}{|@{ } l @{ } | @{ } l @{ } |} \hline
    \adjustbox{valign=t}{
    \parbox{0.4\textwidth}{
    \begin{align*}
      & \text{minimize}_a~ \lefteqn{ \textstyle\sum_{y\in P, x\in C} d^2(y, x)\, a_{yx} } \\
      & \text{subject to }  & \textstyle\sum_{y\in P} a_{yx} & {} = \mu_x & (x\in C) \\
      & & \textstyle\sum_{x\in C} a_{yx}& {} = 1 & (y\in P) \\
      & & a_{yx} & {} \ge 0 & (x\in C, y\in P)
    \end{align*} % \\
    }}
      & % \\
    \adjustbox{valign=t}{
    \parbox{0.5\textwidth}{
    \begin{align*}
      & \text{maximize}_{w,z}~ \lefteqn{ \textstyle\sum_{x\in C} \mu_x \, w_x + \sum_{y\in P} z_y } \\
      & \text{subject to } & z_y & \le d^2(y,x) - w_x & (x\in C, y\in P)
    \end{align*} % \\
    }}
    \\ \hline
  \end{tabular}
\end{center}

This linear program models the standard \emph{transshipment} problem.
As the capacities $\mu_x$ are integers with $\sum_x \mu_x = |P|$,
it is well-known that the basic feasible solutions to the linear program are 0/1 solutions ($a_{yx} \in \{0,1\}$),
and that the (optimal) solutions $a$ correspond to the (minimum-cost) balanced assignments $f:C\rightarrow P$
such that $a_{yx} = 1$ if $f(y) = x$ and $a_{yx} = 0$ otherwise.
%\textcolor{red}
{The implementation here} solves the linear program and dual by
using Goldberg's minimum-cost flow solver~\cite{Goldberg97}
to obtain a minimum-cost balanced assignment $f^*$ and an optimal dual solution $(w^*,z^*)$.
For any minimum-cost balanced assignment $f$ (such as $f^*$)
the resulting weight vector $w^*$ gives a balanced power diagram $\Power C {w^*} f$:

\begin{lemma}[{{see also~\cite{spann_electoral_2007}}}]
  Let $(w^*,z^*)$ b any optimal solution to the dual linear program above.
  Let $f$ be any balanced assignment.
  Then $\Power C {w^*} f$ is a balanced power diagram
  if and only if $f$ is a minimum-cost balanced assignment.
\end{lemma}
\begin{proof}
  Let $a$ be the linear-program solution corresponding to $f$. 

  \smallskip\noindent\emph{(If.)}
  Assume that $f$ has minimum cost among balanced assignments.
  Consider any resident $y\in P$.
  By complimentary slackness, for $x'=f^*(y)$, the dual constraint for $(x', y)$ is tight,
  that is, \(z^*_y  = d^2(y, f(y)) - w^*_{f(y)}. \)
  Combining this with the dual constraint for $y$ and any other $x\in C$ gives
  \[ d^2(y, f(y)) - w^*_{f(y)} \,=\, z^*_y \,\le\, d^2(y, x) - w^*_{x}.\]
  That is, from $y$, the weighted squared distance to $f(y)$
  is no more than the weighted squared distance to any other center $x\in C$.
  So, $y$ is in the power region $C_{f(y)}$ of its assigned center $f(y)$.
  Hence, $f$ is consistent with $\power C w^*$, and $\Power C {w^*} f$ is a balanced power diagram.

  \smallskip\noindent\emph{(Only if.)}
  Assume that $f$ is consistent with $\power C {w^*}$.
  That is, the weighted squared distance from $y$ to $f(y)$
  is no more than the weighted squared distance to any other center $x\in C$.
  That is, defining $z'_y = d^2(y, f(y)) - w^*_{f(y)}$,
  \[ z'_y \,=\, d^2(y, f(y)) - w^*_x\,\le\, d^2(y, x) - w^*_{x}.\]
  Thus, $(w^*, z')$ is a feasible dual solution.
  Furthermore, the complimentary slackness conditions hold for $a$ and $(w^*,z')$.
  That is, $a_{yz} > 0 \implies f(y) = x \implies z'_y = d^2(y, x) - w^*_x$.
  Hence, $a$ and $(w^*, z')$ are optimal.
  Since $a$ is optimal, $f$ has minimum cost.
\end{proof}

%%% Local Variables:
%%% TeX-master: "equitable"
%%% Mode: latex
%%% End:

\subsection{Experiments}

We ran the implementation on various
instances of the redistricting problem.
We considered the following US states: Alabama, California,
Florida, Illinois, New York, and Texas. Note that this list
of states contains the biggest states in terms of population
and number of representatives, and so our algorithm is usually
faster on smaller states.

For each of these states, we used the data provided by
the US Census Bureau~\cite{USCB}, namely the population and
housing unit count by block from the 2010 census. Hence, the input
for our algorithm was a weighted set of points in the plane where
each point represents a block and its weight represents
the number of people living in the block.
For each state, we defined the number of clusters to be
the number of representatives prescribed for the state. See
Table~\ref{T:data} for more details.

\begin{table}[h!]%\label{T:data}
  \centering
  \begin{tabular}{|l|c|c|r|}
    \hline
    State & Number of representatives & Population
    & Number of iterations to converge\\
    \hline
    \hline
    Alabama & 7 & 4779736 & 28\\
    California & 53 & 37253956 & 49\\
    Florida & 27 & 18801310 & 51\\
    Illinois & 18 & 12830632 & 72\\
    New York & 27 & 19378102 & 65\\
    Texas & 36 & 25145561 & 42\\
    \hline
   \end{tabular}
   \caption{The states considered in our experiments together
     with the number of clusters (i.e.:\ number of representatives)
     and number of clients (i.e.\ population of the state).}\label{T:data}
\end{table}

We note that in all cases the
algorithm converged to a local optimum.

\subsection{Technical details and implementation}
{The implementation is} available at \url{https://bitbucket.org/pnklein/district}.
It is written mostly in \texttt{C++}.
Our implementation makes use of a slightly adapted version of a
min-cost flow implementation, \texttt{cs2} due to Andrew Goldberg and
Boris Cherkassky and described in~\cite{Goldberg97}.  The copyright on
\texttt{cs2} is owned by IG Systems, Inc., who grant permission to use for evaluation
purposes provided that proper acknowledgments are given.  If there is
interest, we will write a min-cost flow implementation that is unencumbered.

We also provide Python-3 scripts for reading census-block data,
reading state boundary data, finding the boundaries of the power
regions, and generating \texttt{gnuplot} files to produce the figures
shown in the paper.  These figures superimposed the boundaries of the
power regions and the boundaries of states (obtained from~\cite{USCB2}).

For our experiments, the programs were compiled using \texttt{g++-7}
and run on a laptop with processor \texttt{Intel Core i7--6600U CPU, 2.60GHz} and total virtual
memory of 8GB. The system was \texttt{Debian buster/sid}.
The total running time was less than fifteen minutes for all instances
except California, which took about an hour.

\section{Concluding remarks}

The method explored in this paper outputs districts that are convex
polygons with few sides on average and that are balanced with respect
to population, i.e. where the populations in two districts differ by
at most one.  However, such balance cannot be guaranteed under a
requirement that certain geographical regions, e.g. census blocks or
counties, remain intact.  Since the locations of people within census
blocks are not known, the requirement is sensible.
One possible way to address the requirement is to first compute districts while
disregarding the requirement, then use dynamic programming to modify
the solution to obey that requirement while minimizing the resulting
imbalance.

We have focused in this paper on the Euclidean plane.  This ensures
that each district is the intersection of the geographical region
(e.g.~state) with a polygon.  However, in view of the fact that the
method %\textcolor{red}
{explored here} might generate a district that includes residents
separated by water, mountains, etc.,
one might want to consider
a different metric, e.g.~to take travel time into account.  Suppose,
for example, the metric is that of an undirected graph with
edge-lengths.    One can use essentially the same algorithm for
finding a balanced centroidal power diagram.  Computing a minimum-cost
balanced assignment (Step~1) and the associated weights can still be
done using an algorithm for minimum-cost
flow as described in Section~\ref{duality}.  In Step~2, the algorithm
must move each center to the location that minimizes the sum of
squared distances from the assigned residents to the new center
location.  In a graph, we limit the candidate locations to the
vertices and possibly locations along the edges.  Under such a limit,
it is not hard to compute the best locations.

\section{Acknowledgements}
Thanks to Warren D.~Smith for informing us of references~\cite{olson_rangevoting.org_2011,spann_electoral_2007}.

 \bibliographystyle{plain} 
 \bibliography{Bib/zotero,Bib/other}

\end{document}